\documentclass[11pt]{article}

\usepackage[a4paper,margin=1in]{geometry}
\usepackage{amsmath,amssymb,amsthm,mathtools}
\usepackage{bm}
\usepackage{microtype}
\usepackage[authoryear]{natbib}
\usepackage[hidelinks]{hyperref}
\usepackage[nameinlink,noabbrev]{cleveref}

\newcommand{\dd}{\,\mathrm d}
\newcommand{\E}{\mathbb E}

\newcommand{\Q}{\mathbb Q}
\newcommand{\G}{\mathcal G}

\newcommand{\Ito}{It\^o}

\newtheorem{assumption}{Assumption}[section]
\newtheorem{proposition}[assumption]{Proposition}
\newtheorem{remark}[assumption]{Remark}

\title{\Ito--Wentzell Formula
       and Dupire Stochastic PDE}
\author{Vladimir Lucic\\
\small Department of Mathematics, Imperial College London\\
\small \texttt{vlucic@ic.ac.uk}}
\date{July 2026}

\begin{document}
\maketitle

\begin{abstract}
Starting from the classic result of \citet{Ventzel1965},
we derive a conditional forward equation and an associated
stochastic Dupire PDE for a local-stochastic-volatility
model (LSV).  As an application, we obtain a density-weighted
Rao--Blackwell estimator for the leverage function in LSV.  We also derive
an SPDE for a rolling expiry vanilla option, in the spirit of the Musiela parametrization in interest rate modeling.
\end{abstract}

Rao--Blackwellization.


\section{Model and standing assumptions}

We work under a risk-neutral probability measure $\Q$ on a filtered probability
space supporting two independent Brownian motions $\xi$ and $\eta$.  Set
\[
\G_t:=\sigma(\xi_s:0\leq s\leq t),
\]
completed in the usual way.  Consider the local-stochastic-volatility (LSV)
model
\begin{equation}
\frac{\dd S_t}{S_t}
=b_t\dd t
+L(t,S_t)f(V_t)
\left(\rho\dd\xi_t+\sqrt{1-\rho^2}\dd\eta_t\right),
\qquad b_t=r_t-q_t,
\label{eq:lsv-model}
\end{equation}
where $r$ and $q$ are time-dependent deterministic, $L$ is a deterministic leverage function, and
\begin{equation}
\dd V_t=a(t,V_t)\dd t+\nu(t,V_t)\dd\xi_t.
\label{eq:variance-model}
\end{equation}
We take $S_0>0$ and $V_0\geq0$ to be deterministic.  Thus, under strong
existence for \eqref{eq:variance-model}, $V$ is adapted to the common-noise
filtration $(\G_t)_{t\geq0}$.
The stochastic-volatility model is obtained by taking $L\equiv1$ and
$f(v)=\sqrt v$.

\begin{assumption}\label{ass:regularity}
The coefficients in \cref{eq:lsv-model,eq:variance-model} ensure existence of
a non-explosive strong solution with $S_t>0$ and $V_t\geq0$.  Moreover,
$|\rho|<1$, and the conditional law of $S_t$ given $\G_t$ has a jointly
measurable density $p_t(s)$ for $t>0$.  The coefficients and solution satisfy
the integrability conditions needed for the conditional projections of the
stochastic integrals below.  Strong-form equations are understood under the
additional regularity required for the displayed spatial derivatives;
otherwise they are interpreted weakly.
\end{assumption}

For the call-transform calculations, we additionally assume
$\int_0^\infty s p_t(s)\dd s<\infty$ almost surely and the relevant fluxes
decay at infinity.  More precisely, setting
\[
J_t^\xi(s):=\rho sL(t,s)f(V_t)p_t(s),\quad
Q_t(s):=s^2L^2(t,s)f^2(V_t)p_t(s),\quad
J_t^b(s):=b_tsp_t(s),
\]
we assume, for every $K>0$,
\[
\lim_{s\to\infty}(s-K)J_t^\xi(s)=0,
\qquad
\lim_{s\to\infty}\bigl((s-K)\partial_sQ_t(s)-Q_t(s)\bigr)=0,
\qquad
\lim_{s\to\infty}(s-K)J_t^b(s)=0.
\]
Equivalently, the calculations may be justified by localization and passage
to the limit under corresponding uniform-integrability conditions.

The restriction $|\rho|<1$ is needed to avoid degenerate cases:  if $|\rho|=1$ and $S_0$ is
deterministic, then the conditional law of $S_t$ given $\G_t$ is generally a
Dirac measure.

\section{Conditional forward equation}

For a test function $\varphi\in C_c^2((0,\infty))$, write
\[
\langle p_t,\varphi\rangle
:=\int_0^\infty \varphi(s)p_t(s)\dd s.
\]
The following result is a direct application of the main result of \citet{Ventzel1965}.
\begin{proposition}[Conditional forward equation in the weak form]
Under \cref{ass:regularity}, the conditional density satisfies the weak SPDE
\begin{align}
\dd\langle p_t,\varphi\rangle
&=\left\langle p_t,
b_t s\varphi'(s)
+\frac12s^2L^2(t,s)f^2(V_t)\varphi''(s)
\right\rangle\dd t
\notag\\
&\quad+
\left\langle p_t,
\rho sL(t,s)f(V_t)\varphi'(s)
\right\rangle\dd\xi_t.
\label{eq:weak-forward}
\end{align}
If $p$ is sufficiently regular, then
\begin{align}
\dd p_t(s)
&=-\partial_s\!\left(\rho sL(t,s)f(V_t)p_t(s)\right)\dd\xi_t
\notag\\
&\quad+
\left[
\frac12\partial_{ss}\!\left(s^2L^2(t,s)f^2(V_t)p_t(s)\right)
-\partial_s\!\left(b_tsp_t(s)\right)
\right]\dd t.
\label{eq:strong-forward}
\end{align}
\end{proposition}

\begin{proof}
Apply It\^o's formula to $\varphi(S_t)$ and take the optional projection onto
$(\G_t)_{t\geq0}$.  Under the stated integrability assumptions, the relevant
projection identities are
\[
\E\!\left[\int_0^t Z_u^\eta\dd\eta_u\,\middle|\,\G_t\right]=0,
\qquad
\E\!\left[\int_0^t Z_u^\xi\dd\xi_u\,\middle|\,\G_t\right]
=\int_0^t\E[Z_u^\xi\mid\G_u]\dd\xi_u.
\]
The conditional expectations of the drift and common-noise integrands are
then represented using $p_u$.  This gives \cref{eq:weak-forward}; integration
by parts gives \cref{eq:strong-forward}.
\end{proof}

Intuitively, this is the stochastic forward equation obtained by conditioning
on the entire common-noise path $\xi$, which also determines the volatility
path.  When $\rho=0$, the stochastic transport term disappears and the result
is a pathwise random forward Kolmogorov (Fokker--Planck) PDE.  It becomes deterministic only when
the remaining coefficients are deterministic.

\section{Conditional call transform and stochastic Dupire equation}

Fix a time horizon $T>0$ and define the conditional call transform
\begin{equation}
c_t(K):=\E\!\left[(S_t-K)^+\mid\G_t\right]
=\int_K^\infty(s-K)p_t(s)\dd s, \quad K\geq0,\quad 0<t\leq T,
\label{eq:conditional-call}
\end{equation}
with initial condition $c_0(K)=(S_0-K)^+$.  The density identity at $t=0$
is understood distributionally, with $p_0=\delta_{S_0}$.
Thus $t$ is both the maturity variable and the endpoint of the conditioning
filtration; $c_t$ is not a fixed-calendar-expiry option-price process.  Under
the stated regularity, for $K>0$,
\begin{equation}
\partial_Kc_t(K)=-\int_K^\infty p_t(s)\dd s,
\qquad
\partial_{KK}c_t(K)=p_t(K),
\label{eq:call-derivatives}
\end{equation}
and
\begin{equation}
\int_K^\infty sp_t(s)\dd s=c_t(K)-K\partial_Kc_t(K).
\label{eq:tail-identity}
\end{equation}

\begin{proposition}[Conditional stochastic Dupire equation]
For $K>0$, the maturity-indexed conditional call transform
\eqref{eq:conditional-call} satisfies
\begin{align}
\dd c_t(K)
&=\rho f(V_t)
\left(\int_K^\infty sL(t,s)p_t(s)\dd s\right)\dd\xi_t
\notag\\
&\quad+
\left[
\frac12K^2L^2(t,K)f^2(V_t)\partial_{KK}c_t(K)
+b_t\bigl(c_t(K)-K\partial_Kc_t(K)\bigr)
\right]\dd t.
\label{eq:conditional-dupire}
\end{align}
\end{proposition}

\begin{proof}
First localize $(s-K)^+$ by smooth compactly supported test functions and then
pass to the limit using the stated integrability and boundary conditions.
Equivalently, in the strong formulation, integrate
\cref{eq:strong-forward} against $(s-K)^+$.  For the common-noise term, 
integration by parts gives
\[
-\int_K^\infty(s-K)
\partial_s\!\left(\rho sL(t,s)f(V_t)p_t(s)\right)\dd s
=\rho f(V_t)\int_K^\infty sL(t,s)p_t(s)\dd s.
\]
Two integrations by parts in the second-order term give its value at $s=K$,
and \cref{eq:tail-identity} gives the transport term.
\end{proof}

For $L\equiv1$ and $f(v)=\sqrt v$, \cref{eq:conditional-dupire} reduces to
\begin{align}
\dd c_t(K)
&=\rho\sqrt{V_t}\bigl(c_t(K)-K\partial_Kc_t(K)\bigr)\dd\xi_t
\notag\\
&\quad+
\left[
\frac12K^2V_t\partial_{KK}c_t(K)
+b_t\bigl(c_t(K)-K\partial_Kc_t(K)\bigr)
\right]\dd t.
\label{eq:sv-conditional-dupire}
\end{align}

Note that in the absence of common noise, the above SPDE reduces to the
classical  Dupire forward PDE.

\section{Rao--Blackwellized leverage calibration}
We start by recalling the standard Dupire projection identity for the
leverage function in LSV, originally due to \citet{Lipton2002}.
Let
\[
\bar p_T(K):=\E[p_T(K)]
\]
be a version of the unconditional density of $S_T$.  Conditioning gives
\begin{align}
\E\!\left[f^2(V_T)p_T(K)\right]
&=\bar p_T(K)
\E\!\left[f^2(V_T)\mid S_T=K\right].
\label{eq:density-disintegration}
\end{align}
 Taking expectations in
\cref{eq:weak-forward}, gives 
\begin{equation}
\sigma_{\mathrm{loc}}^2(T,K)
=L^2(T,K)\Theta(T,K),
\qquad
\Theta(T,K):=\E\!\left[f^2(V_T)\mid S_T=K\right].
\label{eq:projection}
\end{equation}

Let $C(T,K)$ be the time-zero discounted market call price surface.  For
deterministic rates and dividend yield, the Dupire local variance is
\begin{equation}
\sigma_{\mathrm{Dup}}^2(T,K)
:=
\frac{2\left[
\partial_TC(T,K)+q_TC(T,K)+(r_T-q_T)K\partial_KC(T,K)
\right]}
{K^2\partial_{KK}C(T,K)}.
\label{eq:dupire-local-variance}
\end{equation}
Matching the LSV marginals to the market
surface gives the standard projection identity
\begin{equation}
L^2(T,K)
=\frac{\sigma_{\mathrm{Dup}}^2(T,K)}
{\E[f^2(V_T)\mid S_T=K]},
\label{eq:leverage-identity}
\end{equation}
completing the derivation of the expression for the leverage function.
 Since its denominator depends on the LSV law, this is a nonlinear
fixed-point relation for the leverage function.

Fix a deterministic candidate leverage function $L$. For the $i$th simulated
common-noise path $\xi^i$, let $V^i$ be the associated volatility path obtained from \eqref{eq:variance-model}, and let
\[
p^i(T,K):=p^{\,\xi^i}_T(K)
\]
be the conditional density obtained by solving \cref{eq:weak-forward} or
\cref{eq:strong-forward} along that same path.

\begin{proposition}[Density-weighted estimator]
Suppose that $0<\E[p_T(K)]<\infty$, and that
$\E[f^2(V_T)p_T(K)]<\infty$.  Then a consistent estimator of
$\Theta(T,K)$ in \cref{eq:projection} is
\begin{equation}
\widehat\Theta_N(T,K)
=\frac{\displaystyle\sum_{i=1}^N
f^2(V_T^i)p^i(T,K)}
{\displaystyle\sum_{i=1}^Np^i(T,K)}.
\label{eq:rb-estimator}
\end{equation}
The corresponding leverage update is
\begin{equation}
\widehat L_N^2(T,K)
=\frac{\sigma_{\mathrm{Dup}}^2(T,K)}
{\widehat\Theta_N(T,K)}.
\label{eq:leverage-update}
\end{equation}
\end{proposition}

\begin{proof}
By \cref{eq:density-disintegration},
\[
\Theta(T,K)
=\frac{\E[f^2(V_T)p_T(K)]}{\E[p_T(K)]}.
\]
The strong law of large numbers applied to the numerator and denominator,
together with $\E[p_T(K)]>0$, gives \cref{eq:rb-estimator}.
\end{proof}

This estimator integrates out the idiosyncratic spot noise by solving a
one-dimensional conditional density equation along each common-noise path.
It is therefore a Rao--Blackwellized alternative
to estimating the level-set conditional
expectation directly from spot particles (e.g.\ \citet[Section~4.2]{RobertCasella2004}).  
In practice, \cref{eq:leverage-update} is applied by time marching, with the
conditional densities and leverage function updated jointly at each $T$.

\section{Fixed-time-to-expiry and the Musiela parametrization}

The conditional Dupire equation derived earlier evolves the object
$c_t(K)=\E[(S_t-K)^+\mid\G_t]$, for which the maturity and conditioning
point coincide.  In this section we focus on a rolling
time to expiry option $\E\!\left[(S_{t+\tau}-K)^+\mid\G_t\right]$, where
$\tau:=T-t$ is the time-to-expiry.  The resulting
change of variables is analogous to the Musiela parametrization of the
forward-rate curve in the HJM model.

The following proposition is a key technical result needed in the sequel,
and is a direct application of the It\^o--Wentzell formula.
\begin{proposition}[Musiela reparametrization]
Let $X_t(T,K)$, $0\leq t\leq T$, be an adapted random field such that, for each
fixed $T$ and $K$,
\begin{equation}
\dd X_t(T,K)=A_t(T,K)\dd t+B_t(T,K)\dd\xi_t.
\label{eq:fixed-expiry-field}
\end{equation}
Assume that the random field has the regularity required by the
It\^o--Wentzell formula and is continuously differentiable in $T$.  Define
\begin{equation}
U_t(\tau,K):=X_t(t+\tau,K),
\qquad \tau\geq0.
\label{eq:musiela-field}
\end{equation}
Then
\begin{align}
\dd U_t(\tau,K)
&=B_t(t+\tau,K)\dd\xi_t
\notag\\
&\quad+
\left[A_t(t+\tau,K)+\partial_\tau U_t(\tau,K)\right]\dd t.
\label{eq:musiela-dynamics}
\end{align}
\end{proposition}

\begin{proof}
The maturity process $T_t=t+\tau$ is of finite variation.  Hence there is no
quadratic-variation correction in the $T$-direction and no cross-variation
between $T$ and $\xi$.  The It\^o--Wentzell formula
\citep[Chapter~3]{Kunita1990} therefore gives
\[
\dd X_t(T_t,K)
=\dd X_t(T,K)|_{T=T_t}
+\partial_TX_t(T_t,K)\dd t,
\]
which is \cref{eq:musiela-dynamics} because
$\partial_TX_t(t+\tau,K)=\partial_\tau U_t(\tau,K)$.
\end{proof}

\begin{proposition}[Rolling conditional option value]
Fix $\tau\geq0$ and $K\geq0$, and suppose
$(S_{t+\tau}-K)^+\in L^2$ for the times under consideration.  Define the
rolling conditional payoff field directly by
\begin{equation}
U_t^{\mathrm u}(\tau,K)
:=\E\!\left[(S_{t+\tau}-K)^+\mid\G_t\right].
\label{eq:rolling-conditional-payoff}
\end{equation}
Assume the random-field regularity needed for the Musiela reparametrization
in the preceding proposition.  Then there is a predictable field
$\beta_t^{\mathrm u}(\tau,K)$ such that
\begin{equation}
\dd U_t^{\mathrm u}(\tau,K)
=\partial_\tau U_t^{\mathrm u}(\tau,K)\dd t
+\beta_t^{\mathrm u}(\tau,K)\dd\xi_t.
\label{eq:rolling-undiscounted-field}
\end{equation}
In particular, $U_t^{\mathrm u}(0,K)=c_t(K)$.

For the corresponding rolling discounted partial-information option value,
let
\[
D(t,t+\tau):=\exp\!\left(-\int_t^{t+\tau}r_u\dd u\right)
\]
and define
\begin{equation}
U_t^P(\tau,K)
:=D(t,t+\tau)U_t^{\mathrm u}(\tau,K)
=\E\!\left[D(t,t+\tau)(S_{t+\tau}-K)^+\mid\G_t\right].
\label{eq:rolling-discounted-value}
\end{equation}
Then
\begin{equation}
\dd U_t^P(\tau,K)
=\left(\partial_\tau U_t^P(\tau,K)+r_tU_t^P(\tau,K)\right)\dd t
+D(t,t+\tau)\beta_t^{\mathrm u}(\tau,K)\dd\xi_t.
\label{eq:rolling-discounted-field}
\end{equation}
\end{proposition}

\begin{proof}
For an auxiliary fixed maturity $T$, define
\[
X_s^{\mathrm u}(T,K)
:=\E\!\left[(S_T-K)^+\mid\G_s\right],
\qquad 0\leq s\leq T.
\]
The tower property makes
$X^{\mathrm u}(T,K)$ a $\G_s$-martingale.  The Brownian martingale
representation theorem applies because $(\G_s)$ is the augmented filtration
generated by $\xi$, so for some predictable $B_s^{\mathrm u}(T,K)$,
\[
\dd X_s^{\mathrm u}(T,K)=B_s^{\mathrm u}(T,K)\dd\xi_s.
\]
Since $U_t^{\mathrm u}(\tau,K)=X_t^{\mathrm u}(t+\tau,K)$, setting
$\beta_t^{\mathrm u}(\tau,K):=B_t^{\mathrm u}(t+\tau,K)$ and applying
\cref{eq:musiela-dynamics} with $A=0$ gives
\eqref{eq:rolling-undiscounted-field}.

For the discounted result, introduce only in the proof the fixed-maturity
field
\[
P_s(T,K):=D(s,T)X_s^{\mathrm u}(T,K),
\qquad
D(s,T):=\exp\!\left(-\int_s^T r_u\dd u\right).
\]
Since $\partial_sD(s,T)=r_sD(s,T)$ for fixed $T$,
\[
\dd P_s(T,K)
=r_sP_s(T,K)\dd s+D(s,T)B_s^{\mathrm u}(T,K)\dd\xi_s.
\]
Now set $s=t$ and $T=t+\tau$ and apply
\cref{eq:musiela-dynamics} once more to obtain
\eqref{eq:rolling-discounted-field}.
\end{proof}

\begin{remark}[Fixed calendar expiry]
Let $T^*$ be fixed and set $\tau_t=T^*-t$.  Applying It\^o--Wentzell to
$U_t(\tau_t,K)$ produces the additional term
$-\partial_\tau U_t(\tau_t,K)\dd t$, which cancels the transport term in
\cref{eq:musiela-dynamics} and recovers \cref{eq:fixed-expiry-field}.
\end{remark}

\section{Conclusion}

The conditional forward SPDE  due to \citet{Ventzel1965} provides a direct route to a stochastic Dupire equation and a density-weighted LSV
calibration procedure.
The result of \citet{Ventzel1965} also serves as a basis for deriving a Musiela-like result for rolling options price processes. The results are applicable to a wide range of LSV models, including the classic stochastic volatility model.

\bibliographystyle{plainnat}
\bibliography{references}

\end{document}